\documentclass{fundam}

\publyear{26}
\papernumber{2}
\volume{196}
\issue{2}
\theDOI{10.46298/fi.15937}

\versionForARXIV

\usepackage{url} 
\usepackage[ruled,lined]{algorithm2e}
\usepackage{graphicx}
\usepackage{tikz}
\usetikzlibrary{automata, positioning, arrows}
\usepackage{amssymb,amsmath}

\newtheorem{observation}{Observation}

\begin{document}
	\title{Paired Disjunctive Domination Number of Middle Graphs}
	
	\address{Dokuz Eylul University, Maritime Faculty, Tinaztepe Campus, 35390, Buca, Izmir, Turkey}
	
	\author{Hande Tuncel Golpek
		\\
		Maritime Faculty \\
		Dokuz Eylul University \\ Tinaztepe Campus, 35390, Buca, Izmir, Turkey\\
		hande.tuncel{@}deu.edu.tr
		\and Zeliha Kartal Yildiz
		\\
		Computer Programing Department  \\
		Izmir Konak Vocational School,
		Izmir Turkey  
\and Aysun Aytac
\\
Department of Mathematics  \\
Ege University,
Izmir Turkey } 
		\maketitle
	
\runninghead{H. Tuncel Golpek, Z. Kartal Yildiz, A. Aytac}{Paired Disjunctive Domination Number of Middle Graphs}
	\begin{abstract}
		The concept of domination in graphs plays a central role in understanding structural properties and applications in network theory. In this study, we focus on the paired disjunctive domination number in the context of middle graphs, a transformation that captures both adjacency and incidence relations of the original graph. We begin by investigating this parameter for middle graphs of several special graph classes, including path graphs, cycle graphs, wheel graphs, complete graphs, complete bipartite graphs, star graphs, friendship graphs, and double star graphs. We then present general results by establishing lower and upper bounds for the paired disjunctive domination number in middle graphs of arbitrary graphs, with particular emphasis on trees. Additionally, we determine the exact value of the parameter for middle graphs obtained through the join operation. These findings contribute to the broader understanding of domination-type parameters in transformed graph structures and offer new insights into their combinatorial behavior.	
	\end{abstract}

	\section{Introduction}
Graph theory, a cornerstone of discrete mathematics, offers a robust framework for modeling and analyzing intricate systems. Its wide-ranging applications include domains such as computer science, biology, transportation, and the social sciences, where graphs provide an intuitive means of representing real-world structures—vertices symbolize individual elements, while edges denote interactions or connections among them. Due to this representational power, graph theory has become a crucial tool in addressing complex challenges, from optimizing communication infrastructures to understanding chemical compound structures.

Within this vast field, domination theory has emerged as a particularly rich area of study. It explores how specific subsets of vertices can exert control or influence over the entire graph. Domination-related parameters play a critical role in both theoretical investigations and practical scenarios, including network defense, facility deployment, and efficient allocation of limited resources. A noteworthy advancement in this area is the paired disjunctive domination parameter—a more recent refinement that introduces an element of redundancy and robustness into domination strategies. This parameter determines the smallest number of vertex pairs needed such that every vertex in the graph either belongs to one of the selected pairs or is adjacent to at least one vertex from a pair. Such a framework is highly relevant for contexts requiring reliability and fault tolerance, including resilient communication networks and strategic emergency response planning.

A practical example of paired disjunctive domination can be observed in urban emergency response systems. Consider a city where ambulance stations must be strategically positioned to ensure prompt medical assistance in case of emergencies. By placing ambulances at designated locations and ensuring that each region is directly served or adjacent to a station, city planners can achieve efficient coverage while maintaining backup support in case of service disruptions. This type of redundancy is crucial in optimizing response times and enhancing overall service reliability.

A fundamental challenge in modern network design is optimizing the balance between resource allocation and redundancy. Since critical resources are often costly, deploying them uniformly across an entire network is impractical. Moreover, the potential failure of resources at specific nodes necessitates redundancy and backup mechanisms, which, while essential for network reliability, impose additional resource demands. This problem has been extensively studied through graph-theoretic models, wherein researchers identify strategically positioned vertex subsets that maintain network connectivity while satisfying redundancy constraints.

Domination theory, including its extensions such as total domination, disjunctive total domination, paired domination, and more recently, paired disjunctive domination, forms a well-established area in graph theory with numerous applications and a rich body of literature \cite{dom1,dom2,totaldom2,totaldom3,HenningPDD}. Foundational surveys by Haynes et al. \cite{dom1,dom2} provide comprehensive insights into these parameters and their theoretical significance.

Given a graph $G=(V,E)$, a subset $S\subseteq V$ is defined as a \textit{dominating set} if every vertex in $V-S$ has at least one neighbor in $S$. \textit{The domination number}, denoted $\gamma(G)$, represents the minimum cardinality of a dominating set in $G$. A subset $S\subseteq V$ is a \textit{total dominating set (TDS)} if every vertex in $V$ is adjacent to at least one vertex in $S$, and \textit{the total domination number}, $\gamma_t(G)$, is the minimum cardinality of a TDS. 

\textit{A matching} in a graph is a collection of pairwise non-adjacent edges, and \textit{a perfect matching} is one that includes every vertex exactly once. \textit{A paired dominating set (PDS)} is a dominating set $S$ for which the subgraph induced by $S$, denoted $G[S]$, has a perfect matching. This ensures that every selected vertex has a designated partner within the set. It is known that any graph without isolated vertices admits a PDS, since the endpoints of a maximal matching satisfy this condition. The paired domination number, $\gamma_{pr}(G)$, represents the minimum size of such a set. This concept was originally introduced by Haynes and Slater \cite{paireddom,paireddom2} in the context of modeling redundancy in security systems, where guards are deployed in mutually supportive pairs.

The concept of \textit{the disjunctive domination} was introduced by Goddard et al. in 2014 as a variation of original domination concept \cite{disjunctivedom}. In a graph $G$, a subset $D\subseteq V$ is called a \textit{b-disjunctive dominating set (bDD-set)} if every vertex not in $D$ either has at least one neighbor in $D$ or has at least $b$ vertices in $D$ at distance two. When $b=2$, the set is referred to as a \textit{disjunctive dominating set (2DD-set)}, and the minimum such set defines \textit{the disjunctive domination number} $\gamma_d(G)$. Various studies have explored algorithmic aspects and bounds of disjunctive domination, including its total version, where the induced subgraph $G[D]$ contains no isolated vertices.
A subset $D\subseteq V$ in an isolate-free graph $G=(V,E)$ is called a \textit{paired disjunctive dominating set (PDD-set)} if it satisfies the following conditions:
\begin{itemize}
\item	D is a disjunctive dominating set, meaning every vertex in $V-D$ has at least one neighbor in D or at least two vertices in D at distance two.
\item	The subgraph $G[D]$ contains a perfect matching, ensuring that every vertex in D is paired with another vertex via an edge.
\end{itemize}

\textit{The paired disjunctive domination number}, denoted as $\gamma_{pr}^d(G)$, represents the minimum cardinality of a \textit{PDD-set} in $G$. A \textit{PDD-set} of cardinality $\gamma_{pr}^d(G)$ is called a $\gamma_{pr}^d$-set \cite{HenningPDD}.

We consider \textit{the middle graph}, a fundamental structural transformation of a given graph that plays a key role in various theoretical and applied settings in this study. The middle graph $M(G)$ of a graph $G$ is obtained by introducing a new vertex for each edge of $G$ and establishing edges between these new vertices whenever their corresponding edges in $G$ share a common vertex. The middle graph $M(G)$ is a well-studied transformation that integrates both the vertex and edge structures of a graph by creating a hybrid representation whose vertex set is $V(G)\cup E(G)$. This construction is not merely a technical modification; it plays a fundamental role in several areas of graph theory. First, middle graphs are closely related to classical graph transformations such as line graphs, total graphs, and splitting graphs, and this relationship enables the transfer of structural and combinatorial properties---such as traversability, connectivity, matching behaviour, and spectral characteristics---across different graph models. Second, middle graphs have been employed in chemical graph theory and the computation of vertex-based topological indices, where the combined vertex--edge representation captures atom--bond interactions more faithfully than standard graph models. Third, a growing body of contemporary research has investigated middle graphs in connection with domination parameters, spectral indices (including nullity and energy), and matching enumeration, demonstrating that they form a rich and flexible framework for analysing refined graph invariants. Within this context, studying the paired disjunctive domination number of middle graphs is a natural and meaningful extension of existing work, as it brings together a fundamental structural transformation with a domination parameter that is sensitive to both local and global connectivity properties \cite{Hamada,Lai,Shirkol,Betul,Kim,Sardara}.

In this paper, we investigate the paired disjunctive domination number of middle graphs. Specifically, we establish lower and upper bounds for this parameter in terms of the order of the graph G. Furthermore, we determined the parameter's value for the middle graph of graphs obtained under certain graph operations. In addition, we explicitly determine the paired disjunctive domination number for various graph families, including special graphs, double star graphs, and friendship graphs. 


The vertex set of middle graph $M(G)$ is defined as $V(M(G))=V(G)\cup %
SD,$ where $V(G)=\{i:1\leq i\leq n\}$ and $SD=\{v_{ij}:$ $%
ij\in E(G)\}$ (or in order to get rid of notational burden in some proofs we can use the set notation $SD=\{u_{t}: t\in {1,\dots, \lvert E(G) \rvert }\}$) denotes the set of new vertices, called as subdivision vertex, corresponding to the edges of $G$. For any two vertices $i,j\in V(G)$, if they are adjacent then the edge $ij\in E(G)$ is represented by a new vertex $v_{ij}$ in the middle graph $M(G).$  Therefore, edge set can be described as $E(M(G))=\{iv_{ij},v_{ij}j:ij\in E(G)\}\cup E(L(G))$, where $L(G)$ is the line graph of $G$ (see in Figure \ref{1}). The line graph $L(G)$ of a graph $G$ is defined so that each vertex of $L(G)$ corresponds to an edge of $G$, and two vertices of $L(G)$ are adjacent whenever their corresponding edges in $G$ share a common endpoint. 
\begin{figure}[ht]
	\centering
	\includegraphics[width=0.6\textwidth]{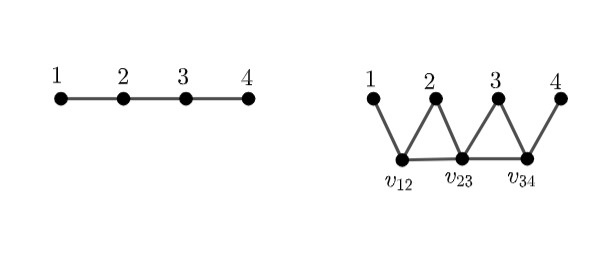}
	\caption{Path graph $P_{4}$ and its middle graph of $M(P_{4})$}
	\label{1}
\end{figure}
\section{Preliminaries}
In this paper, we consider simple, finite, and undirected graphs. The concepts introduced in this section are provided to support the statements of later theorems and to ensure the clarity of the arguments developed throughout the manuscript.

A graph $G=(V,E)$ consists of a nonempty set $V$ of vertices and a set $E$ of unordered pairs of distinct vertices. A graph is called connected if every pair of vertices is joined by a path. A connected graph containing no cycles is referred to as a tree. Trees play a central role in domination-type parameters, and several structural notions related to trees will be used in subsequent sections.
Let $T$ be a tree. A vertex of degree one in $T$ is called a leaf. The set of all leaves of $T$ is denoted by 
$\mathit{Leaf}(T)$, and its cardinality, $\lvert \mathit{Leaf}(T)\rvert$, gives the total number of leaf vertices in the tree. If $v$ is a one-degree  vertex (pendant vertex), then its unique neighbor is called a support vertex. A support vertex adjacent to at least two pendant vertices is designated as a strong support vertex. These notions are particularly important in domination theory and its variants, where leaves and support vertices often determine extremal configurations.

The \emph{open neighborhood} of a vertex $v$ is defined as $N(v)=\{\,u\in V : uv\in E\,\},$
while the \emph{closed neighborhood} of $v$ is given by
$
N[v]=\{v\}\cup N(v).
$
For a general graph $G$, the degree of a vertex $v$, denoted 
$\deg(v)$, is defined as the number of edges incident with 
$v$. Two vertices $u$ and $v$ are said to be adjacent if 
$uv\in E$. Adjacency relations and degree properties will be used repeatedly while establishing local structural constraints in our proofs. We also refer to readers \cite{Chartrand} for all other fundamental graph terms and notations. 

We also recall several specific graph families that will appear in later sections:

The double star graph $D_{n,m}$ where $n \geq m \geq 0$, is the graph obtained by taking the union of two stars $K_{1,n}$ and $K_{1,m}$ and adding an edge between their centers \cite{Grossman}.

A graph obtained by taking $n$ copies of the cycle $C_3$ such that all cycles share a common vertex is called a \emph{friendship graph} \cite{Gallian}. 

The join operation will later be used to describe certain constructions related to our paired disjunctive domination results \cite{Chartrand}.

\section{Known Results}
In this section, we provide some fundamental results related to the paired disjunctive domination parameter, which was introduced by Henning \cite{HenningPDD} and further developed by Golpek and Aytac \cite{Honam, HTGAAShadow}. This will facilitate the flow of the subsequent parts of the paper.

\begin{observation}[\cite{HenningPDD}]\label{Observation1} For an
isolate-free graph $G,$ $\gamma ^{d}(G)\leq \gamma _{pr}^{d}(G)\leq \gamma_{pr}(G)$, and $\gamma ^{d}(G)\leq \gamma _{pr}^{d}(G)\leq 2\gamma
^{d}(G).$\end{observation}

\begin{observation}[\cite{Honam}] For an isolate-free graph $G,$ $%
\gamma _{t}^{d}(G)\leq \gamma _{pr}^{d}(G).$
\end{observation}


\begin{observation}[\cite{HTGAAShadow}]
\label{Shadow2.7} If $T$ is a tree and $D$ is any \textit{PDD-set}, then every support vertex of $T$ is included in $D$ with a its neighbour or $D$ may contains at least two
neighbours of the support vertices with their neighbours.
\end{observation}

\begin{theorem}[\cite{HenningPDD}]
\label{special} For paired disjunctive domination number of some special graphs:
\begin{itemize}
 \item[i)]  $\gamma _{pr}^{d}(C_{n})=2\left\lceil 
	\frac{n}{5}\right\rceil$ with $n\geq 3,$ \\

\item[ii)]  $\gamma _{pr}^{d}(P_{n})=2\left\lceil \frac{n+1}{5}%
	\right\rceil$ with $n\geq 2,$
    \item[iii)]  $\gamma _{pr}^{d}(K_{n})= \gamma _{pr}^{d}(K_{m,n})=2$ with $n\geq 2$.
\end{itemize}
\end{theorem}

\begin{theorem}[\cite{HTGAAShadow}]\label{buyuk2} Let $G$ be isolated-free graph, then $2\leq \gamma _{pr}^{d}(G)\leq n.
	$
\end{theorem}
\section{Results about Special Graph Structures}
In this section, we investigate the paired disjunctive domination parameter in middle graphs of several special graph classes.
\begin{theorem}\label{genelteomyiddle}
	Let $G$ be a connected graph and let $M(G)$ be its middle graph. Suppose that there exists a vertex $u\in V(G)$ and two distinct edges $uv_1,uv_2\in E(G)$ such that for corresponding subdivision vertices $x_1$ and $x_2$ of $M(G)$ , every vertex $z\in V(M(G))\setminus \{x_1,x_2\}$ is either adjacent to at least one of $x_1$ and $x_2$ or satisfies $d_{M(G)}(z,x_1)=d_{M(G)}(z,x_2)=2$. Then $\gamma_{pr}^d(M(G))=2.$
\end{theorem}
\begin{proof}
	Let $x_1$ and $x_2$ be the subdivision vertices corresponding to the edges $uv_1$ and $uv_2$. Since these edges share the vertex $u$, the vertices $x_1$ and $x_2$ are adjacent in $M(G)$ and thus induce a perfect matching. By assumption, every vertex of $M(G)$ outside $\{x_1,x_2\}$ is either adjacent to $x_1$ or $x_2$ or is at distance two from both, which means that $\{x_1,x_2\}$ is a paired disjunctive dominating set of $M(G)$. Therefore, $\gamma_{pr}^d\leq 2$ is obtained. Furthermore, by Theorem~2.2, it is known that the paired disjunctive domination number is at least \(2\).
	That is,
	$
	\gamma_{pr}^d(M(G)) \ge 2.
	$
	Consequently, we conclude that
	$
	\gamma_{pr}^d(M(G)) = 2.
	$
\end{proof}

\begin{proposition}\label{n-1dereceli}
	Let $G$ be a graph on $n$ vertices without isolated vertices. 
	Suppose the maximum degree of $G$ satisfies $\Delta(G) = n-1$, then
	\[
	\gamma_{pr}^d(M(G)) = 2.
	\]
\end{proposition}
\begin{proof}
	Let $u$ be a vertex of $G$ with $\deg_G(u)=n-1$, and choose two distinct neighbors 
	$v_1$ and $v_2$ of $u$. Denote by $x_1$ and $x_2$ the subdivision vertices of $M(G)$ 
	corresponding to the edges $uv_1$ and $uv_2$, respectively. Using the fact that $u$ is 
	adjacent to all other vertices of $G$, one easily verifies that every vertex of $M(G)$ 
	is either adjacent to at least one of $x_1$ or $x_2$, or is at distance two from both 
	of them. Hence, by Theorem \ref{genelteomyiddle}, we obtain
	\[
	\gamma_{pr}^d(M(G)) = 2.
	\]
\end{proof}


  \begin{proposition}
	Let $K_{m,n}$ be complete bipartite graph with $m+n$ vertices where $m,n \geq 1$. Then $$\gamma_{pr}^d(M(K_{m,n}))=2.$$
\end{proposition}
\begin{proof}
	Let $u=x_1$ and consider the edges $x_1y_1$ and $x_1y_2$ in $K_{m,n}$ with corresponding subdivision vertices $a$ and $b$ in $M(K_{m,n})$. It is easy to see that every vertex is either adjacent to $a$ or $b$ or at distance two from both of them. Hence by Theorem \ref{genelteomyiddle}, we obtain $
	\gamma_{pr}^d(M(K_{m,n})) = 2.
	$
\end{proof}

\begin{theorem}
    Let $C_n$ be the cycle and $n\geq3$. Then, $$\gamma_{pr}^d(M(C_n))=2\lceil\frac{n}{4}\rceil.$$
\end{theorem}
\begin{proof}
    Let us denote the vertex set of the middle graph, $M(C_n)$ as the union of two disjoint subsets; vertices of $C_{n}$ and subdivision vertices $SD$: 
\[
V_1 = \{i \mid 1 \leq i \leq n,\, i \in V(C_n)\}, \quad SD = \{u_{j}\mid 1 \leq j \leq n\},
\]so that
\[
V(M(C_n)) = V_1 \cup SD.
\]
We first construct an upper bound for $\gamma_{pr}^d(M(C_n))$. Define the set 

\begin{align*} D =\begin{cases} \{u_j \mid 1 \leq j \leq n,\ j \equiv 1  \text{ or } 2\pmod{4}\},\ n\not\equiv 1\pmod{4}\\
\{u_j \mid 1 \leq j \leq n,\ j \equiv 1 \text{ or } 2\pmod{4}\}\cup \{n\},\ n\equiv 1\pmod{4}.
\end{cases}
\end{align*}

Since every vertex in $M(C_n)$ is either in $D$ or disjunctively dominated by a vertex in $D$, and since the subgraph induced by $D$ contains a perfect matching, it follows that $D$ is a \textit{PDD-set} of $M(C_n)$ for all cases of $n$ modulo 4. 
Furthermore, we have $ |D| = 2\left\lceil \frac{n}{4} \right\rceil$. Hence,

\begin{equation}\label{cycleust}
    \gamma_{pr}^d(M(C_n)) \leq |D| = 2\left\lceil \frac{n}{4} \right\rceil.
\end{equation}
Now, we establish a lower bound for $\gamma_{pr}^d(M(C_n))$. Let $D'$ be a $\gamma_{pr}^d$-set of $M(C_n)$. Consider the matching in subgraph $M(C_n)[D']$ given by 
\[
M = \{u_{i_1}u_{i_1+1}, u_{i_2}u_{i_2+1}, \dots, u_{i_r}u_{i_r+1}\},
\]
where $1 \leq i_1 < i_1+1 < i_2 < i_2+1 < \dots < i_r < i_r+1 \leq n$, and for each $k\in\{1,\dots,r\}$, the vertices $u_{i_{k}}$ and $u_{i_{k+1}}$ form a vertex pair.

Notice that between any two consecutive matched pairs $u_{i_k}u_{i_k+1}$ and $u_{i_{k+1}}u_{i_{k+1}+1}$, there are at most two subdivision vertices that are not in $D'$ for any $k \in \{1, \dots, r\}$. Therefore, we obtain
\[ n \leq 2r + 2r = 4r, \]
which implies that \begin{equation}\label{cyclealt}
    |M| = 2r \geq 2\left\lceil \frac{n}{4} \right\rceil.
\end{equation}
Combining the upper bound (\ref{cycleust}) and lower bound (\ref{cyclealt}), we conclude that $\gamma_{pr}^d(M(C_n)) = 2\left\lceil \frac{n}{4} \right\rceil. 
$
\end{proof}
\begin{theorem}
    Let $P_n$ be the path with $n$ vertices, $n\geq2$. Then, $$\gamma_{pr}^d(M(P_n))=\begin{cases}
        2\lceil\frac{n-1}{4}\rceil, &n\equiv2,3 \pmod{4}\\
         2\lceil\frac{n-1}{4}\rceil+2, &\text{otherwise.}
    \end{cases}$$
\end{theorem}
    \begin{proof}
Let us denote the vertex set of the middle graph $M(P_n)$ as the union of two disjoint subsets: \[ V_1 = \{i \mid 1 \leq i \leq n,\, i \in V(P_n)\}, \quad SD = \{u_j \mid 1 \leq j \leq n-1\}, \]so that
\[
V(M(P_n)) = V_1 \cup SD.
\] We first construct an upper bound for $\gamma_{pr}^d(M(P_n))$. Define the set 
\[D = \{u_j \mid 1 \leq j \leq n-1,\ j \equiv 1,2\pmod{4}\}.
\]

Now, consider the following cases:
\begin{itemize}
 \item If $n \equiv 0,1 \pmod{4}$, let $S = D \cup \{u_{n-1}, v_n\}$.
\item  If $n \equiv 2 \pmod{4}$, let $S = D \cup \{v_n\}$.
 \item If $n \equiv 3 \pmod{4}$, then $S = D$.
 \end{itemize}

In all cases, $S$ forms a \textit{PDD-set} in $M(P_n)$. Furthermore, we observe that If $n \equiv 2,3 \pmod{4}$, then $|S| = |D| = 2\left\lceil \frac{n-1}{4} \right\rceil$. Otherwise, $|S| = |D| + 2 = 2\left\lceil \frac{n-1}{4} \right\rceil + 2$. Thus, we obtain the following upper bound:
\begin{equation}\label{pathust}
    \gamma_{pr}^d(M(P_n)) \leq |S| = 
\begin{cases}
2\left\lceil \frac{n-1}{4} \right\rceil, & \text{if } n \equiv 2,3 \pmod{4}, \\
2\left\lceil \frac{n-1}{4} \right\rceil + 2, & \text{otherwise}.
\end{cases}
\end{equation}

Now we construct a lower bound for $\gamma_{pr}^d(M(P_n))$. Let $D'$ be a $\gamma_{pr}^d$-set in $M(P_n)$. Consider the induced subgraph $M(P_n)[D']$ and let $M = \{u_{i_1}u_{i_1+1}, u_{i_2}u_{i_2+1}, \dots, u_{i_r}u_{i_r+1}\}$ be the matching in $M(P_n)[D']$, where 
\[1 \leq i_1 < i_1 + 1 < i_2 < i_2 + 1 < \dots < i_r < i_r + 1 \leq n-1.
\]

Two cases arise:
\begin{itemize}
\item Case 1: If $n \equiv 2,3 \pmod{4}$.

Note that between any two consecutive matched pairs $u_{i_k}u_{i_k+1}$ and $u_{i_{k+1}}u_{i_{k+1}+1}$, there are at most two subdivision vertices that are not in $D'$, for each $k \in \{1,\dots,r\}$. Hence, we obtain 
$n - 1 \leq 2r+2r=4r$ and therefore \begin{equation}\label{pathalt1}
    |M| = 2r \geq 2\left\lceil \frac{n-1}{4} \right\rceil.
\end{equation} 

\item Case 2: If $n \equiv 0,1 \pmod{4}$.

In this case, let $N = M \cup \{v_n, u_{n-1}\}$. By Case 1, we know that $|M| = 2r$, and therefore \begin{equation}\label{pathalt2}
    |N| = 2r + 2 \geq 2\left\lceil \frac{n-1}{4} \right\rceil + 2.
\end{equation}
\end{itemize}
Combining the lower bounds obtained in both cases (\ref{pathalt1}) and (\ref{pathalt2}) with the previously established upper bounds (\ref{pathust}), we conclude that 
$\gamma_{pr}^d(M(P_n))=
\begin{cases}
2\left\lceil \frac{n-1}{4} \right\rceil, & \text{if } n \equiv 2,3 \pmod{4}, \\
2\left\lceil \frac{n-1}{4} \right\rceil + 2, & \text{otherwise}.
\end{cases}   
$
\end{proof}
\begin{proposition}
	Let $F_n$ be friendship graph with $n\geq 2$. Then $$\gamma_{pr}^d(M(F_n))=2.$$
\end{proposition}
\begin{proof}
	Let $v_1$ be the central vertex of $F_n$ and consider one of the triangles $\{v_1,v_2,v_3\}$. In the middle graph $M(F_n)$, let $x_1$ and $x_2$ be the subdivision vertices corresponding to the edges $v_1v_2$ and $v_1v_3$, respectively. Then $x_1$ and $x_2$ are adjacent in $M(F_n)$ and, it follows immediately that every vertex of $M(F_n)$ is either adjacent to at least one of $x_1$ or $x_2$, or at distance two from both of them. Hence by Theorem \ref{genelteomyiddle}, we obtain $
	\gamma_{pr}^d(M(F_n)) = 2.
	$
\end{proof}
%
%
\begin{theorem}
    Let $D_{n,m}$ be the double star graph where $n,m \geq 1$. Then $$\gamma_{pr}^d(M(D_{n,m}))=4.$$
\end{theorem}
    \begin{proof}
    Let $M(D_{n,m})$ denote the middle graph of the double star graph $D_{n,m}$. Let the vertex set of $M(D_{n,m})$ be partitioned as $V(M(D_{n,m})) = V_1 \cup V_2$, where 
\[
V_1 = \{i \mid 1 \leq i \leq n+m+2,\ i \in V(D_{n,m})\},  SD = \{u_j \mid 1 \leq j \leq n+m+1\}.
\]

We know from Theorem~\ref{buyuk2} that $\gamma_{pr}^d(G) \geq 2$ for any graph $G$. Thus, $\gamma_{pr}^d(M(D_{n,m})) \geq 2$. 
        To complete the proof, we will show that no set of two vertices can form a \textit{PDD-set} in $M(D_{n,m})$, and therefore $\gamma_{pr}^d(M(D_{n,m})) \geq 4$.

Assume, for contradiction, that there exists a \textit{PDD-set} $S'$ of $M(D_{n,m})$ with $|S'| = 2$. Since $S'$ must induce a perfect matching and ensure that every vertex in $M(D_{n,m})$ is either in $S'$ or disjunctively dominated by a vertex in $S'$, the only possibility is that $S'$ consists of two adjacent vertices. That is, $S' = \{x, y\}$ such that $d_{M(D_{n,m})}(x, y) = 1$.

However, due to the structure of the double star graph, the vertices $x$ and $y$ together can disjunctively dominate at most a localized region of the graph either within a single star component or among a small subset of adjacent vertices. This is insufficient to cover the entire graph, especially when $n, m \geq 1$, since the graph contains at least $n + m + 3$ vertices in total. Thus, not all vertices are disjunctively dominated by $S'$, and $S'$ cannot be a \textit{PDD-set}. Hence, no set of two vertices can form a \textit{PDD-set} in $M(D_{n,m})$, implying
\[
\gamma_{pr}^d(M(D_{n,m})) \geq 4.
\]
In the case of $M(D_{n,m})$, we construct a specific set of four vertices that forms a \textit{PDD-set} and achieves this bound. Let the two central vertices of the stars in $D_{n,m}$ be denoted by $a$ and $b$, and let $u_a$ and $u_b$ be any two edge-subdivision vertices adjacent to $a$ and $b$, respectively in $M(D_{n,m})$. Let
$
S = \{a, u_a, b, u_b\}
$ is a PDD-set.  By construction, every vertex in $M(D_{n,m})$ is either in $S$ or disjunctively dominated by a vertex in $S$, and the subgraph induced by $S$ contains a perfect matching. Therefore, $S$ is a \textit{PDD-set} of $M(D_{n,m})$, and we have
\[
\gamma_{pr}^d(M(D_{n,m})) \leq 4.
\]
Therefore, we conclude that
$\gamma_{pr}^d(M(D_{n,m})) = 4.$
    \end{proof}

\section{Main Results}
In this section, we present general results on the paired disjunctive domination of middle graphs. In particular, we establish lower and upper bounds for trees and general graphs, and provide exact value related to the join operation.
\begin{lemma}\label{Lemma1}
Let $G$ be a connected graph with $n\geq 5$ vertices and $S$ be a \textit{PDD-set} of $M(G).$ Then there exist $S^{\prime }\subset SD$ a \textit{PDD-set} of $M(G)$ with $\left\vert
S^{\prime }\right\vert \leq \left\vert S\right\vert .$
\end{lemma}
\begin{proof}
    If \( S \subset SD \), then the result is satisfied by choosing \( S' = S \). Otherwise, the set \( S \) also contains vertices from \( V(G) \). In this case, let us select a vertex \( t \in S \cap V(G) \). If all vertices \( v_{it} \in SD \) (or \( v_{ti} \)) adjacent to \( t \) are in \( S \), then two cases arise:\\
Case 1: If the vertex \( t \) does not contribute to the paired property of the set \( S \), then we take \( S_{1} = S - \{t\} \).\\
Case 2: If the vertex \( t \) contributes to the paired property of the set \( S \), then there are two subcases:
\begin{itemize}
\item If all vertices \( v_{it} \in SD \) adjacent to \( t \) are in \( S \), then we select a vertex \( v_{mn} \in SD \setminus S \) from the set \( N_{2}(v_{it}, M(G)) \), and take \( S_{1} = (S \cup \{v_{mn}\}) \setminus \{t\} \).
\item If not all adjacent \( v_{it} \in SD \) are in \( S \), suppose \( v_{it} \in M \setminus S \). In this case, we take \( S_{1} = (S \cup \{v_{it}\}) \setminus \{t\} \).
\end{itemize}
This process terminates in a finite number of steps, and eventually, a set \( S^{'} \) consisting only of vertices in \( SD \) is obtained. Furthermore, in each step, \( S^{'} \) remains a \textit{PDD-set}. Therefore, \( S^{'} \) is a \textit{PDD-set} of M(G). Thus, we construct $\left\vert
S^{\prime }\right\vert \leq \left\vert S\right\vert .$
\end{proof}

\begin{lemma}
Let $G$ be a connected graph with $n\geq 5$ vertices. Let $t\in V(G)$ be a vertex that is distinct from a
support vertex. Then,%
\[
\gamma _{pr}^{d}(M(G-t))\leq \gamma _{pr}^{d}(M(G))\leq \gamma
_{pr}^{d}(M(G-t))+2.
\]
\end{lemma}
\begin{proof}
    (i) First, let us show that $\gamma_{pr}^{d}(M(G))\leq\gamma_{pr}^{d}(M(G-t))+2$. Let \( S \) be a \textit{PDD-set} of the graph \( M(G - t) \). In this case, for every \( w \in N_{\leq2}(t, G) \), we must have either \( w \in S \) or \( v_{iw} \in SD \), which implies that \( v_{iw} \in S \). As a result, for some \( w_1, w_2 \in N_G(t) \), the set \( S \cup \{ v_{w_1t}, v_{tw_2} \} \) is a \textit{PDD-set} for \( M(G) \). Thus, we obtain the conclusion 
    \[\gamma_{pr}^{d}(M(G))\leq\gamma_{pr}^{d}(M(G-t))+2.\]
(ii) Now, let us show that $\gamma_{pr}^{d}(M(G-t))\leq\gamma_{pr}^{d}(M(G))$. Let $S$ be a $\gamma_{pr}^{d}-set$ of $M(G)$. From Lemma \ref{Lemma1}, we assume that $S \subseteq SD$. Let us consider the set 
\[S_{t} = N_{\leq2}(t, M(G)) \cap S. \] Since $S$ is a $\gamma_{pr}^{d}-set$, it is easy to see that \( |S_t| \geq 1 \). Let the neighbors and the vertices at distance $2$ from vertex $t$ in the graph $G$ be denoted by \( w \in N_G(t) \) and \( p \in N_2(t, G) \), respectively. Assume that \( |S_t| = k \), where \( k \geq 1 \), and the vertices in \( S_t \) are of the form \( v_{wt} \) or \( v_{pw} \). Based on these types of vertices, two cases arise:

\noindent
Case 1: If \( v_{wt} \in S \), then the set \[ S_{1} = (S - v_{wt}) \cup \{w\} \] is a \textit{PDD-set} for the graph \( M(G - t) \).

\noindent
Case 2: If \( v_{pw} \in S_t \), then the set \( S_{1} = (S - v_{pw}) \cup \{p\} \) is a \textit{PDD-set} for the graph \( M(G - t) \).
For each vertex in $S_{t}$, the corresponding operation described above is applied depending on which of the two cases it falls under, resulting in a \textit{PDD-set} \( S_k \) for the graph \( M(G - t) \). It is easy to verify that \( |S_k|= |S| \). Hence, it follows that $\gamma_{pr}^{d}(M(G-t))\leq\gamma_{pr}^{d}(M(G)).$ 
\end{proof}

\begin{theorem}
Let $G$ be a connected graph of order $n\geq2.$ Then, 
\[
2\leq \gamma _{pr}^{d}(M(G))\leq \gamma _{pr}^{d}(M(P_{n}).
\]
\end{theorem}
\begin{proof}
According to Theorem \ref{buyuk2}, the lower bound is obtained as
\[
\gamma_{pr}^{d}(M(G)) \geq 2.
\]

Let $G$ be a spanning tree, and let $S$ be a \( \gamma_{pr}^{d} \)-set of $M(G)$. By Lemma \ref{Lemma1}, we may assume that \( S \subseteq SD \). The most vulnerable structure is the path graph considering the number of connections. This implies,
\[
|S| \leq \gamma_{pr}^{d}(M(P_n)).
\]
Since \( S \) is a \( \gamma_{pr}^{d} \)-set of \( M(G) \), it follows that
\[
\gamma_{pr}^{d}(M(G)) \leq \gamma_{pr}^{d}(M(P_n)).
\]
Thus, the desired inequality is obtained.
\end{proof}

\begin{theorem}
\label{tree}Let $T$ be a tree graph, which is not a star. Let $%
SSV=\{v_{i}:1\leq i\leq k\}$ be the strong support vertices of the tree $T.$
For any $u,w\in \mathrm{leaf}(T)$, under the conditions that $N_{G}(u)\cap
N_{G}(w)=\emptyset $ and $d_{G}(u,w)>3,$ it follows that%
\[
\gamma _{pr}^{d}(M(T))\geq 2k+2(\left\vert \mathrm{leaf}(T)\right\vert
-\sum\limits_{i=1}^{k}(\deg (v_{i})-1)).
\]
\end{theorem}
\begin{proof}
Let the set of leaves of the tree \( T \) be denoted by \( \mathrm{leaf}(T) \). The set \( \mathrm{leaf}(T) \) contains strong support vertices. Let us denote these vertices by the set\\ 
\[ SSV = \{v_1, \dots, v_k\}. \]In this case, the number of support vertices in the graph is given by  
\[
|\mathrm{leaf}(T)| - \sum_{i=1}^{k}(\deg(v_i) - 1).
\]

The vertices labeled as \( e_{iv_j} \), where \( i = 1, \dots, \deg_T(v_j) \) in \( M(T) \) corresponding to the edges incident to each vertex \( v_j \), for \( j = 1, \dots, k \), form a complete graph with $\deg_G(v_j)$ vertices.
Let \( S \) be a \textit{PDD-set} of \( M(T) \). Then, by Theorem \ref{special}(iii), it is easy to observe that \( |S| \geq 2k \).

If the graph contains a support vertex, then by Observation \ref{Shadow2.7}, the set \( S \) must include at least one support vertex. Moreover, to maintain the paired property, the neighbors of these support vertices must also be included in \( S \). Hence, the following inequality is obtained
\[\gamma _{pr}^{d}(M(T))\geq 2k+2(\left\vert \mathrm{leaf}(T)\right\vert
-\sum\limits_{i=1}^{k}(\deg (v_{i})-1)).
\]
\end{proof}
\begin{corollary}
Let $T$ be a tree, distinct from a star graph, and containing   no support vertices. Then 
\[
\gamma _{pr}^{d}(M(T))\geq 2k.
\]
where $k$ is the number of the strong support vertices. Equality holds for $diam(T)=4$.
\end{corollary}
\begin{proof}
	From Theorem \ref{tree}, its known that 
	\[\gamma _{pr}^{d}(M(T))\geq 2k+2(\left\vert \mathrm{leaf}(T)\right\vert
	-\sum\limits_{i=1}^{k}(\deg (v_{i})-1)).\] If any support vertex is not included by $T$, then \[\left\vert \mathrm{leaf}(T)\right\vert
	=\sum\limits_{i=1}^{k}(\deg (v_{i})-1).\] Hence, \[
	\gamma _{pr}^{d}(M(T))\geq 2k.
	\]
	If $diam(T)=4$, then for each strong support vertex, by selecting the two subdivision vertices adjacent to it, every vertex of the middle graph is disjunctively dominated. Since these two vertices are adjacent, the paired condition is also satisfied. Therefore, for  $diam(T)=4$, we have obtained that $\gamma _{pr}^{d}(M(T))= 2k.$
	\end{proof}

\begin{corollary}
Let $T$ be a tree without strong support vertices, then 
\[
\gamma _{pr}^{d}(M(T))\geq 2(\left\vert \mathrm{leaf}(T)\right\vert .
\]
If $diam(T)=4$, then equality holds.
\end{corollary}
\begin{proof}
	From Theorem \ref{tree}, its known that 
	\[\gamma _{pr}^{d}(M(T))\geq 2k+2(\left\vert \mathrm{leaf}(T)\right\vert
	-\sum\limits_{i=1}^{k}(\deg (v_{i})-1))\] where $k$ is number of strong support vertices. If any strong support vertices is not included by $T$, then  \[
	\gamma _{pr}^{d}(M(T))\geq 2(\left\vert \mathrm{leaf}(T)\right\vert
	\] is yield.
	
	If $diam(T)=4$, then for each leaf vertex, by selecting the a subdivision vertex adjacent to it, every vertex of the middle graph is disjunctively dominated. To satisfy the paired condition, a vertex adjacent to each selected vertex must be included in the \textit{PDD-set}. Therefore, for  $diam(T)=4$, we have obtained that $\gamma _{pr}^{d}(M(T))= 2(\left\vert \mathrm{leaf}(T)\right\vert).$
\end{proof}

\begin{proposition}
	Let $G$ and $H$ be two graphs and $G+H$ be join of them. Then
	\[
	\gamma_{pr}^d(M(G+H)) = 2.
	\]
\end{proposition}

\begin{proof}
	Let $u\in V(G)$ and choose two distinct vertices $h_1,h_2\in V(H)$. 
	In the join $G+H$, the edges $uh_1$ and $uh_2$ exist, and let $x_1$ and $x_2$ be their corresponding subdivision vertices in $M(G+H)$. Clearly, $x_1$ and $x_2$ are adjacent in $M(G+H)$ since the edges $uh_1$ and $uh_2$ share the endpoint $u$. Moreover, using that every vertex of $G$ is adjacent to every vertex of $H$ in the join, it is straigthforward to verify that every vertex of $M(G+H)$ is either adjacent to $x_1$ or $x_2$, or at distance two from both of them.  Hence by Theorem \ref{genelteomyiddle}, we obtain 
	$\gamma_{pr}^d(M(G+H))=2$.
\end{proof}

%
%
%
%

\end{document}